\documentclass[letter,11pt]{article}
\pdfoutput=1
\usepackage{amsmath,amssymb, amsthm}
\usepackage{float}
\usepackage{algorithm}
\floatname{algorithm}{Protocol}
\usepackage{multirow}
\usepackage{graphicx}
\usepackage{amssymb,amsmath}
\usepackage[margin=1in]{geometry}
\usepackage[normalem]{ulem}

\usepackage{color}

\newcounter{linenum}

\newcommand{\ket}[1]{\left\lvert #1 \right\rangle}
\newcommand{\bra}[1]{\left\langle #1 \right\lvert}

\DeclareMathOperator{\Tr}{\operatorname{Tr}}

\newcommand{\AR}[2][c]{$$\begin{array}[#1]{lllllllllllllll}#2\end{array}$$}

\def\EQ#1{\begin{eqnarray}#1\end{eqnarray}}

\usepackage{color}

\newtheorem{theorem}{Theorem}
\newtheorem{lemma}{Lemma}

\newtheorem{definition}{Definition}

\usepackage[affil-it]{authblk}

\usepackage{hyperref}
\usepackage{natbib}
\usepackage{braket}
\usepackage{xspace}

\begin{document}

\allowdisplaybreaks[3]
\frenchspacing

\title{Optimised resource construction for verifiable quantum computation}

\author[1,2]{Elham Kashefi}
\author[1]{Petros Wallden}
\affil[1]{School of Informatics, University of Edinburgh,}
\affil[ ]{10 Crichton Street, Edinburgh EH8 9AB, UK}
\affil[2]{CNRS LTCI, Departement Informatique et Reseaux,}
\affil[ ]{Telecom ParisTech, Paris CEDEX 13, France}

\hyphenation{avenues}
\hyphenation{research}
\hyphenation{interact}
\hyphenation{machine}
\hyphenation{machines}
\hyphenation{giving}
\hyphenation{paper}
\hyphenation{encoun-ter}
\hyphenation{encoun-ters}
\hyphenation{expe-ri-ence}
\hyphenation{analyse}
\hyphenation{analysis}
\hyphenation{remains}
\hyphenation{logical}
\hyphenation{cannot}

\clubpenalty 10000
\widowpenalty 10000

\date{}

\maketitle

\begin{abstract}

Recent developments make the possibility of achieving scalable quantum networks and quantum devices closer. From the computational point of view these emerging technologies become relevant when they are no longer classically simulatable. Hence a pressing challenge is the construction of practical methods to verify the correctness of the outcome produced by universal or non-universal quantum devices. A promising approach that has been extensively explored is the scheme of verification via encryption through blind quantum computing initiated by Fitzsimons and Kashefi. We present here a new construction that simplifies the required resources for any such verifiable blind quantum computing protocol. We obtain an overhead that is linear in the size of the input, while the security parameter remains independent of the size of the computation and can be made exponentially small. Furthermore our construction is generic and could be applied to any non-universal scheme with a given underlying graph.  
\end{abstract}

\section{Introduction}

It is widely believed that quantum computers and generally quantum devices, can outperform their classical counterparts. In particular, there are problems that it is believed a quantum computer could solve efficiently, while a classical computer would require exponentially (in the size of the input) long time.  In general, it is not possible to classically simulate such computations and therefore in order to verify that a generic quantum device functions correctly, we need to resort to different techniques. Currently the most efficient ways to verify a quantum computation, is to employ cryptographic methods, where we have an almost classical verifier that executes the computation using an untrusted but fully quantum prover. There has been a number of such verification methods \cite{fk,abe,efk,kdk2015,KKD14,tomo2014,hm2015,ruv2,mckague,gkw2015,hpf2015,DFPR13,Broadbent2015,GHZ_verification,MM2015,exp_secret,BGS13} where generally there exists a trade-off between the practicality of the scheme versus their generality, trust assumptions and security level. It is the target of this work to both reduce the experimental requirements of the most general schemes and to achieve further improvements in the more restricted schemes. In order to make quantum verification schemes practical number of aspects can be considered: 

\begin{itemize}
\item whether the verifier's devices are trusted or not;
\item which are the verifier's quantum technological requirements;
\item is the scheme suitable for any universal computation or for only a restricted class;
\item is the output of the quantum computation classical or quantum;
\item does the protocol tolerates errors due to noise; 
\item how does the probability of failure scale;
\item which is the classical and quantum overhead and which is the round complexity of the scheme; 
\item whether the quantum communication needs to be done during the computation (online) or can it be done at some earlier stage (offline);
\end{itemize}
A full review of all necessary parameters are beyond the scope of this paper but it is worth noticing that all the above aspects have been addressed using protocols based on verification via blind quantum computing \cite{fk,efk,kdk2015,KKD14,gkw2015,hpf2015,DFPR13}. We will refer to this family of protocols collectively as VBQC schemes where the key idea is based on hiding the underlying computation (also known as blindness). This would allow the verifier to encode simple trap computations within a general computation that runs on a remote device in such a way that the computation is not affected, while revealing no information to the device. The correctness of the general computation is then tested via the verification of the trap computation. The latter is significantly less costly and thus leads to an efficient scheme (essentially similar to an error detection code). What makes the procedure work is the blindness that hides the trap computation from the actual one. To elaborate further on the security parameter scaling, consider the following informal definition of \emph{verification} that we formalise later (for details see also \cite{fk}).

\begin{definition}\label{epsilon-verification}
A quantum computation protocol is $\epsilon$-verifiable if the probability of accepting an incorrect output (classical or quantum) for any choice of the prover's strategy is bounded by $\epsilon$.
\end{definition}
In a practical scenario, to be convinced of the correctness of the output obtained from a given quantum device, one needs a verification protocol where the security parameter ($\epsilon$) can be made arbitrarily small while keeping the cost (in terms of the experimental requirements) optimal. The standard technique for amplification when dealing with classical output is to simply repeat the protocol multiple times (let say $d$) and if all rounds are accepted and result in the same outputs, then this output is the correct except with probability $\epsilon^d$. However, dealing with quantum output requires more elaborate methods (to deal with the possibility of coherent attacks) 
that involves the use of full fault-tolerant computation and the presence of  multiple traps in order to achieve exponential bounds.

\subsection{Our Contribution}

In this work we focus to further improve the underlying resource construction required for VBQC schemes. Our main results can be summarised as follows:

\begin{enumerate}

\item In Section \ref{Sec:DTG}, inspired by the dotted-complete graph state introduced in \cite{fk}, we give a generic construction where for any given (universal or non-universal graph state resource) multiple trap qubits isolated from the computation qubits 
can be added. Unlike the dotted-complete graph state the overhead of the new construction is only linear in the size of the specific computation that will be performed. Furthermore the traps are uniformly distributed and their positions are essentially independent from each other. 

\item We use this construction to obtain a new universal VBQC protocol (Section \ref{Sec:Verification1}) that has lower cost. Since we are using a different resource, the proof technique had to be accordingly adapted. Our protocol before embedding any boosting mechanism will already have a constant security parameter and thus allows a potentially straightforward one-shot experiment.

\item When the output of the quantum computation is classical, we use a repetition technique to boost the security of our protocol to arbitrarily small $\epsilon$ (Section \ref{Sec:Repetition}). Importantly, we can achieve this using a constant number of repetitions that is independent of the size of the computation. In previous VBQC protocols the number of repetitions that were required increased with the size of the computation. 

\item For the general quantum output case, we use a fault-tolerant encoding of the computation and this boosts the security to arbitrary small $\epsilon$ while in the same time we still require only linear, in the size of the computation, overhead (Section \ref{Sec:FT_Verification}). The overhead of most of previous VBQC protocols are quadratic in the size of the computation with the only exception of \cite{kdk2015}. 
\end{enumerate}

Our generic construction could be explored to optimise various other existing VBQC and we briefly comment on that in Section \ref{Sec:Examples}.

\subsection{Related works}

There has been a number of papers on verification addressing different aspects. With no aim to give a complete list we give here a brief description of some related works. Aharonov, Ben-Or and Eban \cite{abe} provided the first verification protocol. It requires a linear overhead in the size of the computation, but also requires a verifier that has involved quantum abilities, and in particular that can prepare entangled states of size that depends on the security parameter.  

Following another approach, based on measurement-based quantum computation, Fitzsimons and Kashefi \cite{fk} obtained the most optimal scheme from the point of view of verifier’s capability.  
However, the overhead of the full scheme becomes quadratic. Recently a solution for addressing this issue was proposed in \cite{kdk2015} by combining the above two protocols \cite{fk} and \cite{abe} in order to construct a hybrid scheme. This was the only verification protocol (before our work) that requires linear number of qubits while in the same time requires that the verifier has the minimal quantum property of preparing single quantum systems. However, the protocol requires the preparation of qudits (rather than qubits) where the dimension is dictated by the desired level of security. Moreover the required resource is still constructed based upon the dotted-complete graph state though of small constant size. Hence further investigation is required for establishing the experimental simplicity of the two schemes, ours and the one in \cite{kdk2015}.

The first experimental implementation of a simplified verification protocol was presented in \cite{efk} where a repetition technique was explored as well. Other experiments on verifiable protocols include \cite{exp_secret} and an experiment based on the protocol in \cite{GHZ_verification}. However, none of these works are applicable to a full universal scheme like ours. 

On the other hand to achieve a classical verifier new techniques are proposed at the cost of increasing the overall overhead of the protocol dramatically \cite{ruv2} or increasing the number of the provers \cite{mckague}. Other device-independent protocols \cite{gkw2015,hpf2015} used a single universal quantum prover and an untrusted qubit measuring device and while the complexity improved it is still far from experimentally realisable.

The VBQC protocol could be generally viewed as prepare and send scheme (using the terminology from QKD). Equivalent schemes based on measurement-only could also be obtained \cite{tomo2014,hm2015}. In this scenario the prover prepares a universal resource and sends it qubit-by-qubit to the verifier that performs different measurements in order to complete a quantum computation. These protocols are referred to as online protocols (in contrast to the offline protocols mentioned above) since the quantum operations of the verifier occur when they know what they wants to compute. The online scheme can also achieve verification either by creating traps \cite{tomo2014}, or by measuring the stabiliser of the resource state \cite{hm2015}. These protocols could be improved using our techniques as we will comment in Section \ref{Sec:Examples}.

Finally a composable definition of \cite{fk} is given in \cite{DFPR13}, while a limited computational model (one-pure-qubit) is examined in \cite{KKD14}. Due to the generic nature of our construction these results would be applicable to our protocol as well.  

The verification protocols in \cite{Broadbent2015,BGS13} are teleportation based. However, due to the existence of a general mapping between the teleportation  (with two qubits measurement) and one-way computing (with one qubit measurement), see details in \cite{Aliferis04,mbqc}, one should be able to explore any possible improvement that our techniques could bring to these protocols, i.e. extending them to a fault-tolerant setting with quantum outputs.

\subsection{Background}

The family of VBQC protocols are conveniently presented in the measurement-based quantum computation (MBQC) model
\cite{onewaycomputer} that is known to be the same as any gate teleportation model \cite{childs2005unified}. We will assume that the
reader is familiar with this model (also known as the one-way model), whereas further details can be found in \cite{mbqc}. The general idea
behind an MBQC protocol is that one starts with a large and highly entangled multiparty state (the resource state) and the computation is
performed by carrying out single-qubit measurements. There is a (partial) order on the measurements since the basis of a measurement
could depend on the outcomes of previous measurements. The resource states used are known as \emph{graph states} as they could be fully determined by a given graph see details in \cite{hein2004multiparty}. One physical way to construct a graph state given the graph
description is to assign to each vertex of the graph a qubit initially prepared in the state $\ket{+}$ and for each edge of the graph to perform a $\textrm{controlled-}Z$ gate to the two adjacent vertices.

If one starts with a graph state that the qubits were prepared in a rotated basis $\ket{+_\theta}=1/\sqrt{2}(\ket{0}+e^{i\theta}\ket{1})$
instead, then it is possible to perform the same computation with the non-rotated graph state by preforming measurements in a similarly
rotated basis. This observation led to the formulation of the \emph{universal blind quantum computation} (UBQC) protocol \cite{bfk}.
Here a client prepares rotated qubits, where the rotation is only known to them. Client sends the qubits to the server (as soon as they
are prepared hence there is no need for any quantum memory). Finally the client instructs the server to perform entangling operations according to the graph and perform single qubits measurements in suitable angles and order to perform the desired computation. During this protocol the client receives  the intermediate outcomes of previous measurements allowing them to classically evaluate the next measurement angle. Due to the unknown rotation the server does not learn what computation they actually perform.

The UBQC protocol could be turned into a verification protocol where the client (referred to now as verifier) could detect a cheating
server (referred to now as prover). To do so, the verifier for certain vertices (called dummies) sends states from the set
$\{\ket{0},\ket{1}\}$ which has the same effect as a $Z$-basis measurement on that vertex. In any graph state if a vertex is measured
in the $Z$-basis it results in a new graph where that vertex and all its adjacent edges are removed. During the protocol the prover does not
know for a particular vertex if the verifier send a dummy qubit or not. This enables the verifier to isolate some qubits (disentangled
from the rest of the graph) and those qubits have fixed deterministic outcomes if the prover followed honestly the
instructions. The positions of those isolated qubits are unknown to the prover and the verifier uses them as traps to test that the prover
performs the quantum operations that is given. This technique resulted in the first universal VBQC protocol \cite{fk} which is the basis of
our paper. As we explain later, while the trapification idea is straightforward however it is challenging to find the most optimal way
of inserting trap qubits while not breaking the general computation. This is the central focus of this paper to introduce a general
optimised scheme for constructing graph state resources for VBQC protocols.

\section{The dotted triple-graph construction}\label{Sec:DTG}

Our construction starts with a ``base'' graph $G$ such that the related graph state $\ket{G}$ can be used as the resource to perform a particular (or universal) quantum computation in MBQC. This graph is then ``decorated'' in a suitable way, resulting to a graph that we will call dotted triple-graph $DT(G)$. This new graph leads to a suitable resource state $\ket{DT(G)}$ for running a verified quantum computation in an efficient way. The general idea is to construct the $DT(G)$ graph which after some operations (chosen secretly by the verifier) can be broken to three identical graphs. The one will be used to perform the desired computation and the other two to insert trap computations to detect possible deviations. The way that the $DT(G)$ is broken is chosen by the verifier and thus the prover is blind about which vertex belongs to which graph. This general idea was first introduced in \cite{fk}. The key difference of our construction is that while in \cite{fk} the breaking to subgraphs occurs in a global way, in our construction it happens locally. This difference results in a reduction on the number of vertices (and thus qubits) that are required. The precise meaning of ``locality'' will become apparent later after we introduce our construction. 

In this section we will only give definitions and properties of the dotted triple-graph construction when viewed purely as graph operations. Those properties will play a crucial role in the next sections where we will use as resource state the dotted triple-graph state $\ket{DT(G)}$ in order to obtain verifiable quantum computation protocols. Before going to the construction we give a definition:

\begin{definition}\label{dotting} We define the \emph{dotting} operator $D$ on graph $G$ to be the operator which transforms a graph $G$ to a new graph denoted as $D(G)$ and called \emph{dotted} graph, by replacing every edge in $G$ with a new vertex connected to the two vertices originally joined by that edge. We call the set of vertices of $D(G)$ previously inherited from $G$ as primary vertices $P(D(G))$, and the vertices added by the $D$ operation as added vertices denoted by $A(D(G))$.  
\end{definition}

\noindent\textbf{Dotted triple-graph construction:} 
\begin{enumerate}

\item We are given a base graph $G$ that has vertices $v\in V(G)$ and edges $e\in E(G)$, as in Figure \ref{figure1} (a). In the following steps we will give the new graph $DT(G)$, called dotted-triple graph and specify its vertices and edges.

\item For each vertex $v_i$, we define a set of three new vertices $P_{v_i}=\{p^{v_i}_1,p^{v_i}_3,p^{v_i}_3\}$. 

\item Corresponding to each edge $e(v_i,v_j)\in E(G)$ of the base graph that connects the base vertices $v_i$ and $v_j$, we introduce a set of nine edges $E_{e(v_i,v_j)}$ that connect each of the vertices in the set $P_{v_i}$ with each of the vertices in the set $P_{v_j}$.

\item The graph that its vertices are $\cup_{v_i\in V(G)}P_{v_i}$ and the edges are defined as in the previous step is called triple-graph $T(G)$, as in Figure \ref{figure1} (b).

\item We perform the dotting operator $D$ on the triple graph $T(G)$ to obtain the dotted triple-graph $DT(G)$. An example of dotted triple-graph can be seen in Figure \ref{figure1} (c).

\end{enumerate}

\begin{figure}
\includegraphics[width=1.0\columnwidth]{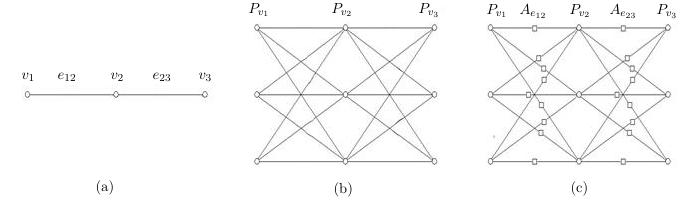}
\caption{(a) A base graph consisting of three vertices and two edges. (b) A triple-graph $T(G)$ where for each vertex $v$ there is a set of three vertices $P_v$. (c) A dotted triple-graph. For each edge of the base graph there is a set of nine added vertices $A_e$. The added vertices are denoted as squares, while the primary as circles.}
\label{figure1}
\end{figure}

Note that, according to Definition \ref{dotting} and the labeling in the above construction the primary vertices are given as  $P(DT(G)) = \cup_{v_i} P_{v_i}$. For convenience we also label the added vertices $A(DT(G))$ as follows. Corresponding to each edge $e(v_i,v_j)$ of the base graph $G$, there are now 9 added vertices and we will denote each set of added vertices as $A_{e_{ij}}=\{a^{e_{ij}}_1,\cdots,a^{e_{ij}}_9\}$. Note that the number of vertices of the new graph is $|V(DT(G))|=3|V(G)|+9|E(G)|$.  
If the maximum degree of the base graph is a constant $c$ then the number of vertices of the $DT(G)$ are linear in the number of vertices of the base graph.  
This property means that if we can base our verifiable quantum computation protocol on this graph, then the number of qubits we will need is linear in the size of the computation. 

Having defined the dotted triple-graph construction, we give some definitions before we prove certain properties of the $DT(G)$. Note that in what follows we present a coloring scheme that is not in effect a graph coloring of the $DT(G)$. The next definition is essentially a labeling scheme that for convenience we present it as a coloring. Therefore connected vertices could get the same color.

\begin{definition}[Trap-Colouring]\label{trap colouring} We define trap-colouring to be an assignment of one colour to each of the vertices of the dotted triple-graph that is consistent with the following conditions. 
\begin{enumerate}
\item[(i)] Primary vertices are coloured in one of the three colours white, black or green. 
\item[(ii)] Added vertices are coloured in one of the four colours white, black, green or red. 
\item[(iii)] In each primary set $P_v$ there is exactly one vertex of each colour. 
\item[(iv)] Colouring the primary vertices fixes the colours of the added vertices: Added vertices that connect primary vertices of different colour are red. Added vertices that connect primary vertices of the same colour get that colour.
\end{enumerate}
\end{definition}
Note that the choice of colours for each of the primary sets $P_v$ can be chosen randomly and is independent from the choices made on other primary sets.  We can also see that in each of the added sets $A_e$ we have one white, one black, one green and six red vertices.

\begin{figure}\label{figure2}
\includegraphics[width=1.0\columnwidth]{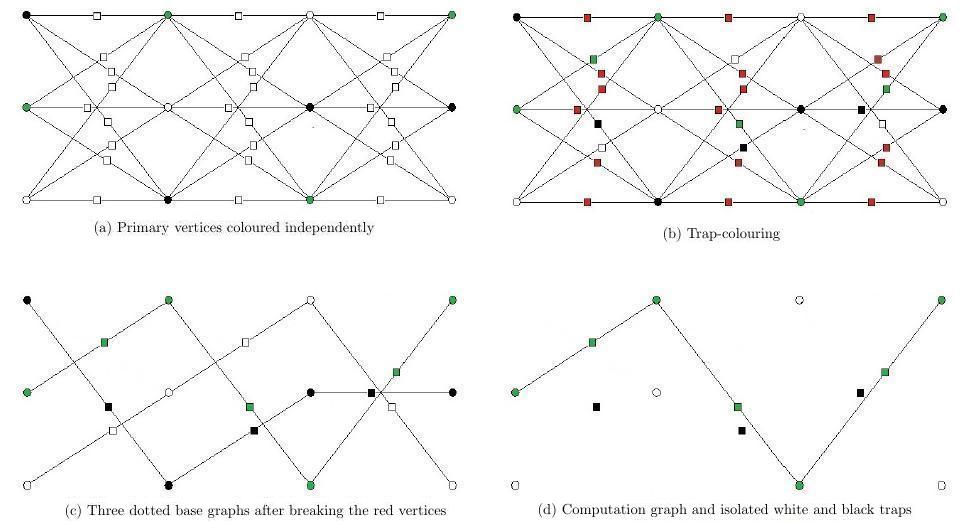}
\caption{(a) A dotted triple-graph, where only the primary vertices are coloured, and this is done randomly for each set. (b) A trap-colouring of $DT(G)$ that is fully fixed from the colouring of the primary vertices. (c) $DT(G)$ after performing break operations on all red vertices. This results to three copies of the dotted base graph. (d) $DT(G)$ after performing further break operations on the primary vertices of the black graph and added vertices of the white graph. The result is a dotted base graph (green) and isolated white traps on primary vertices and black traps on added vertices.}
\end{figure}

In Figure \ref{figure2} (a) and (b) we see an example of trap-colouring, where in (a) we choose independently the colour choices of primary vertices and in (b) the colours of added vertices is fixed following the rules for trap-colouring given above. Before proving the next property we present a graph operation (also introduced in \cite{fk}).

\begin{definition} 
We define the \emph{break} operator on a vertex $v$ of a graph $G$ to be the operator which removes vertex $v$ and any adjacent edges to $v$ from $G$. 
\end{definition}

\begin{lemma}\label{break_red}
Given the dotted triple-graph $DT(G)$ and a trap-colouring, by performing break operations on the red vertices we obtain three identical copies of the dotted base graph $D(G)$ each of them consisting of a single colour. 
\end{lemma}

\begin{proof}
First we note that after the break operations on red added vertices, all the vertices of different colours are disconnected. This follows since edges connecting different colour primary vertices were coloured by definition red, while all added vertices that were not red are connected with same colour vertices. Then, we need to show that the graph of each colour results to a graph identical to $D(G)$. To see this note that for each vertex $v_i$ of the base graph, there is a white (black, green) vertex in $P_{v_i}$. Then for each edge $e(v_i,v_j)$ of the base graph $G$, there is a unique white (black, green) added vertex in $A_{e_{ij}}$ that joins the white vertex $p_w^{v_i}\in P_{v_i}$ and the white vertex $p_w^{v_j}\in P_{v_j}$ (and similarly for black and green). 
\end{proof}
Figure \ref{figure2} (c) illustrates how the $DT(G)$ breaks to three identical dotted base graph, after performing break operations on the red vertices. 

\begin{definition} 
We define the \emph{base-location} of a vertex $f\in V(DT(G))$ of the dotted triple-graph to be the position that the set $P_v$ or $A_e$ that includes $f$ has in the dotted base graph $D(G)$. This position is denoted by either ``$v$'' corresponding to the primary vertex of $D(G)$ or with ``$e$'' corresponding to the added vertex of $D(G)$ on the edge $e$.
\end{definition}
Given a trap-colouring, each primary vertex belongs to one of the three graphs where the colour is determined by the trap-colouring. However, its base-location is fixed prior to the trap-colouring. Here we can see the difference of our construction with that of \cite{fk}. There a dotted-complete graph was used and the graph also broke to three identical graphs, where all primary vertices belonged to one of these graphs. However, there was no restriction of how this break happens, and any choice of three equal subsets was valid. In our construction we maintain the structure of the base-location (reducing the number of added vertices required), but in the same time the colour choices at one primary base-location are totally independent from colour choices at other primary base-location.

The next property that was essentially proved in \cite{fk}, is also required for the verification protocols presented in the next sections. 

\begin{lemma}\label{break_traps}
Given a dotted graph $D(G)$, by applying break operators to every vertex in $P(D(G))$ or $A(D(G))$ the resulting graph is composed of the vertices of $A(D(G))$ or $P(D(G))$ respectively and contains no edges.
\end{lemma}

\begin{proof}
As the dotting operation only introduces vertices connected to vertices in $P(D(G))$, every vertex in $A(D(G))$ shares edges only with vertices in $P(D(G))$. Thus when the vertices in $P(D(G))$ and their associated edges are removed by the break operators, the vertices in $A(D(G))$ become disconnected. Similarly, since the dotting operation removes all edges between vertices in $P(D(G))$, hence every vertex in $P(D(G))$ shares edges only with vertices in $A(D(G))$. Thus when the vertices in $A(D(G))$ and their associated edges are removed by the break operators, the vertices in $P(D(G))$ become disconnected.
\end{proof}

In Figure \ref{figure2} (d) we see that after the break operations of Firgure \ref{figure2} (c), further break operations are performed on all white added vertices and on all black primary vertices. We end up with a (green) copy of the dotted base graph and white isolated traps at primary vertices and black isolated traps at added vertices. 

There are common properties that we will prove for both primary and added vertices and for the ease of notation we will refer to either such set $P_v$ or $A_e$ as $F_l$ with the convention that the subscript $l$ denotes the base-location of the set and when it takes value $v$ (primary base-location) it becomes $P_v$ and when it takes value $e$ (added base-location) it becomes $A_e$. 

Next we show that while the trap-colouring is a global construction it can indeed be considered as a local scheme. This property will be explored in our proof technique for the verification. We formalise this notion in the next set of definitions and lemmas. 

\begin{definition}
We define \emph{local-colouring} of a set $F_l$ to be an assignment of colours to that set that is consistent with some global trap-colouring.
\end{definition}
This definition captures the idea of colouring a particular set $F_l$ corresponding to base-location $l$ such that it can be part of some global trap-colouring without having \emph{any} further constraints from colours of vertices at other base-locations. We can see that a local-colouring of an added set $A_{e_{ij}}$ fully determines the colours of the vertices in the two neighbouring primary base-locations $P_{v_i},P_{v_j}$, while the converse is also true. A local-colouring of the two primary sets $P_{v_i},P_{v_j}$ fully determines the colours of the added set $A_{e_{ij}}$. We can therefore see that a local-colouring of set $A_{e_{ij}}$ is equivalent with a local-colouring of the two neighboring primary base-locations $P_{v_i},P_{v_j}$. We can also see that a local-colouring of \emph{all} primary sets $P_v$ is compatible with a trap-colouring and fixes it uniquely.

However, if we have two general sets $F_{l_1},F_{l_2}$ it is not always possible to colour them both using a local-colouring and still be able to find a global trap-colouring. An example is if we have a primary set $P_{v_1}$ and its added neighbouring set $A_{e_{1k}}$, where a local colouring of the set $P_{v_1}$ imposes constraints on the colours of $A_{e_{1k}}$ further than those required from a local-colouring. To see this note that an added vertex connected to a white primary vertex can be either white or red, but can never be black. 
One can see that an added set $A_{e_{ij}}$ can have local-colouring if there is no constraint on the colours from the neighbouring primary sets $P_{v_i},P_{v_j}$, but also from other added sets $A_{e_{ik}},A_{e_{jk}}$ that have common neighbor sets (either $P_{v_i}$ or $P_{v_j})$. Here we wish to make precise when there is a collection of base-locations that one can assign (independently) local-colourings to all the related sets $F_l$ and still be able to always find a global trap-colouring.

\begin{definition}[Independently Colourable Locations (ICL)]\label{ICL} 
Given a dotted triple-graph $DT(G)$ and a collection of $n$ base-locations $\mathcal{E}$ with corresponding sets $F_l$, we call the set $\mathcal{E}$ \emph{independently colourable locations} if any local-colouring  of the sets $F_l$ is consistent with at least one trap-colouring.
\end{definition}
What this definition captures is that the choice of colours within each of the sets $F_l$ corresponding to a base-location in $\mathcal{E}$ is independent from the choice of colours in other sets $F_{l'}$ with base-location in $\mathcal{E}$. 

For each base-location $l$ we define $\epsilon_{l}=\{l\}$ if the base-location is primary and $\epsilon_{l}=N_{D(G)}(l)$ if the base-location is added (i.e. in the latter case, it contains the two primary base-locations that are adjacent to the location $l$).
\begin{lemma} \label{ICL3}
A set of $n$ base-locations $\mathcal{E}$ is ICL if and only if for all pairs $i,j\in \mathcal{E}$ the sets $\epsilon_{i}\cap\epsilon_{j}=\emptyset$. 
\end{lemma}

\begin{proof}
First we prove that a collection of base-locations satisfying this condition, is ICL. From $\epsilon_i\cap\epsilon_j=\emptyset$ we can see that (i) for all primary base-locations in $\mathcal{E}$ no neighbouring base-location is in $\mathcal{E}$ and (ii) for each added base-location, the two neighbouring primary-locations $P_{v_i},P_{v_j}$ are not in $\mathcal{E}$ and neither is any other added base-location set that has neighbours either of $P_{v_i},P_{v_j}$. In other words, the sets of neighbours of added base-locations are disjoint. However, we already noted that a local-colouring of an added base-location is equivalent with a local-colouring of the two neighboring primary base-locations $P_{v_i},P_{v_j}$. By replacing the local-colouring of added base-locations with that of the neighbouring primary base-locations, we reduce the local-colouring of the set $\mathcal{E}$ to that of a collection of local-colourings of primary base-locations. This is ICL since by the definition of trap-colouring no constraint is imposed between the colours of primary sets.

To prove the converse consider two locations $i,j$ such that $\epsilon_i\cap\epsilon_j\neq\emptyset$. Either one is primary and the other is a neighbouring added base-location or $i$ and $j$ are added base-locations sharing a common neighbour $k$. In the first case it is clear that the choice of colour at the primary set (say $i$) imposes constraints on the colours of the added base-location $j$. In the second case, the choise of colour at the added location $i$ can determine that of the neighbour location $k$ (for example a white added vertex that is connected with a primary vertex fixes the colour of that vertex to white). But then fixing the colours of the primary base-location $k$ in its turn imposes constraints for the other added neighbour $j$, and thus a local-colouring of $i$ and $j$ may not be consistent with a trap-colouring. 
\end{proof}

We are now in position to prove the following property that is necessary for Section \ref{Sec:FT_Verification}.

\begin{theorem}\label{ICL2}
Consider a dotted triple-graph $DT(G)$. Consider a set $S$ of $n$ base-locations and assume that the base graph $G$ has maximum degree $c$. Then there exists a subset $S'\subseteq S$ of these base-locations that are ICL (independent colourable locations) and it contains at least  $|S'|=\frac{n}{2c+1}$ locations. 
\end{theorem}

\begin{proof}
The set $S$ has $n$ locations of the graph $D(G)$. We want a subset of these locations $S'$ such that it satisfies Lemma \ref{ICL3}. The condition of that lemma requires that if a primary base-location $v_i$ is included, then all its neighbouring base-locations should be excluded. The maximum number of neighbours is given by the maximum number of added base-locations  which is $c$.  Therefore if we include the base-location $v_i$ in the set $S'$, we might need to exclude at most $c$ other base-locations from the set $S$.

To include any added base-location $e_{ij}$ in the set $S'$, Lemma \ref{ICL3} requires that its neighbours $v_i,v_j$ and the neighbours of its neighbours $e_{ik},e_{jk}$ should be excluded. Its neighbours are 2, while the neighbours of the neighbours are at most $2(c-1)$. It follows that to include $e_{ij}$ in the set $S'$ we might need to exclude at most $2c$ other base-locations from the set $S$.

From the pigeonhole principle follows that we can find a set $S'$ with at least $\frac{n}{2c+1}=|S'|$ base-locations that are ICL.
\end{proof}

\section{Verifiable quantum computation}\label{Sec:Verification1}

In this section we give a verifiable blind quantum computation protocol using the dotted triple-graph construction, but otherwise, we follow similar steps with \cite{fk}. 
With our construction we obtain a protocol where the probability of success is constant  (independent of the size of the computation) and we use only linear, in the size of the computation, number of qubits.

\begin{algorithm}[H]
\caption{Verifiable Universal Blind Quantum Computation using dotted triple-graph}
 \label{prot:AUBQC}
We assume that a standard labelling of the vertices of the dotted triple-graph $DT(G)$ used, is known to both the verifier and the prover. The number of qubits is at most $3N(3c+1)$ where $c$ is the maximum degree of the base graph $G$.\\
\noindent$\bullet$ \textbf{Verifier's resources} \\
-- Verifier is given a base graph $G$ that the dotted graph state $\ket{D(G)}$ can be used to perform the desired computation in MBQC with measurement pattern $\mathbb{M}_{\textrm{Comp}}$.\\
-- Verifier generates the dotted triple-graph $DT(G)$, and selects a trap-colouring according to definition \ref{trap colouring} which is done by choosing independently the colours for each set $P_v$.\\
-- Verifier for all red vertices will send dummy qubits and thus performs break operation.\\
-- Verifier chooses the green graph to perform the computation.\\ 
-- Verifier for the white graph sends dummy qubits for all added qubits $a^{e}_w$ and thus generates white isolated qubits one at each primary vertex set $P_{v}$. Similarly for the black graph the verifier sends dummy qubits for the primary qubits $p^v_b$ and thus generates black isolated qubits one at each added vertex set $A_{e}$.\\
\noindent -- The dummy qubits position set $D$ is chosen as defined above (fixed by the trap-colouring).\\
\noindent -- A sequence of measurement angles, $\phi=(\phi_i)_{1\leq i \leq 3N(3c+1)}$ with $\phi_i \in A=\{0,\pi/4,\cdots,7\pi/4\}$, consistent with $\mathbb{M}_{\textrm{Comp}}$, where $\phi_i = 0$ for all the trap and dummy qubits. 
The verifier chooses a measurement order on the dotted base-graph $D(G)$ that is consistent with the flow of the computation (this is known to prover). The measurements within each set $P_v,A_e$ of $DT(G)$ graph are order randomly.
\\
\noindent -- $3N (3c+1)$ random variables $\theta_i$ with value taken uniformly at random from $A$.\\
\noindent -- $3N (3c+1)$ random variables $r_i$ and $|D|$ random variable $d_i$ with values taken uniformly at random from $\{0,1\}$. \\
\noindent -- A fixed function $C(i, \phi_i, \theta_i, r_i, \mathbf{s})$ that for each non-output qubit $i$ computes the angle of the measurement of qubit $i$ to be sent to the prover.

\noindent$\bullet$ \textbf{Initial Step} \\
-- \textbf{Verifier's move:} Verifier sets all the value in $\mathbf{s}$ to be $0$ and prepares the input qubits as
\AR{
\ket e = X^{x_1} Z(\theta_1) \otimes \ldots \otimes  X^{x_l} Z(\theta_l) \ket I
}
and the remaining qubits in the following form
\AR{
\forall i\in D &\;\;\;& \ket {d_i} \\ 
\forall i \not \in D &\;\;\;& \prod_{j\in N_G(i) \cap D} Z^{d_j}\ket {+_{\theta_i}}
}
and sends the prover all the $3N (3c+1)$ qubits in the order of the labelling of the vertices of the graph.

-- \textbf{Prover's move:} Prover receives $3N(3c+1)$ single qubits and entangles them according to $DT(G)$.

\noindent$\bullet$ \textbf{Step $i: \; 1 \leq i \leq 3N (3c+1)$}

-- \textbf{Verifier's move:} Verifier computes the angle $\delta_i=C(i, \phi_i, \theta_i, r_i, \mathbf{s})$ and sends it to the prover.\\ 
-- \textbf{Prover's move:} Prover measures qubit $i$ with angle $\delta_i$ and sends the verifier the result $b_i$. \\
-- \textbf{Verifier's move:} Verifier sets the value of $s_i$ in $\mathbf{s}$ to be $b_i+r_i$. \\

\noindent$\bullet$  \label{step:Alice-prep} \textbf{Verification} \\
-- After obtaining the output qubits from the prover, the verifier measures the output trap qubits with angle $\delta_t=\theta_t+r_t\pi$ to obtain $b_t$.

-- Verifier accepts if $b_i = r_i$ for all the white (primary) and black (added) trap qubits $i$.

\end{algorithm}

As it is evident from the protocol, the positions of the dummy qubits (i.e. those that are $\{\ket{0},\ket{1}\}$) is determined by the trap-colouring. It is known that sending dummy qubits has the same effect as making a $Z$ measurement in MBQC which effectively breaks the graph state at this vertex. Therefore the properties defined in Section \ref{Sec:DTG} corresponding to the reduction of the $DT(G)$ to one dotted base graph $D(G)$ and isolated traps (Lemmas \ref{break_red} and \ref{break_traps}) as well as the properties concerning the independence of the colouring and thus the distribution of traps (Theorem \ref{ICL2}), all apply here. 

\begin{theorem} (correctness)
If both verifier and prover follow the steps of protocol \ref{prot:AUBQC} then the output is correct and the computation accepted.
\end{theorem}

\begin{proof}
If both verifier and prover follow the steps of protocol \ref{prot:AUBQC} then the prover essentially (when pre-rotations are taken into account) applies the pattern $\mathbb{M}_{\textrm{Comp}}$ at the green dotted base graph $D(G)$, which by assumption performs the desired computation (see also theorems 1 and 3 of \cite{fk}). Moreover, the isolated white and black qubits are measured in the correct basis and thus the verifier receives $b_i=r_i$ for the traps and accepts the computation.
\end{proof}

\begin{theorem} (Verification)\label{Verification1}
Protocol \ref{prot:AUBQC} is $\left(\frac{8}{9}\right)$-verifiable both when the output is quantum or classical.
\end{theorem}

\begin{proof}
The proof follows closely steps of the proof of theorem 8 of ref \cite{fk} and thus here we will only outline the steps of the proof and stress where we differ, while a more detailed proof can be found in Appendix \ref{App:proof1}. 

The proof consists of five steps. In \textbf{step 1} we prove that any deviation from the ideal protocol can be expressed in terms of some Kraus operators which are then written as linear combination of strings of Pauli matrices (denoted as $\sigma_i$) and the remaining of the proof is to see which of those attacks maximise the probability of accepting an incorrect outcome. 

In \textbf{step 2} we note that there are some strings $\sigma_i$ that for \emph{any} choice of the secret parameters (trap positions, angles, etc) of the verifier do not corrupt the computation and thus they do not contribute to the probability of failure. The set of all the other strings $\sigma$ (that could corrupt the computation for some choice of parameters) will be denoted as $E_i$. It is clear that the prover, to optimise the chance to get an incorrect outcome accepted, should only use attacks from the set $E_i$. In this section, where we consider the simplest protocol, a single non-trivial attack could corrupt the computation and the set $E_i$ consists of all the attacks $\sigma$'s that have in at least one position a non-trivial attack. However, in the next section this changes. The technique to amplify the success probability uses fault-tolerant encoding and thus the computation is corrupted only if multiple errors occur and this leads to different set $E_i$. For now we keep the description general for as long as possible, so that it applies for the next section. After the set $E_i$ is defined, in order to compute an upper bound for the failure probability, we simply compute the probability of not triggering any trap given that the attacks used are all from the set $E_i$. This is clearly an upper bound for the failure probability (worse-case scenario), since in reality the fact that there exist some choices of the secret parameters that a given $\sigma\in E_i$ corrupts the computation does not mean that it corrupts the computation in general. However, an upper bound $\epsilon$ of the failure probability suffices to prove that the protocol is $\epsilon$-verifiable\footnote{It is worth pointing out, that since this is a weak bound a different proof technique could result to a tighter bound for the failure probability.}.

In \textbf{step 3} we exploit the blindness of the prover. The fact that the prover does not know the secret parameters restricts the attacks that contribute to the bound of the failure probability we compute to a convex combination of Pauli attacks. This is important since it eliminates ``coherent'' type of attacks and resembles theorems in quantum key distribution (QKD) that reduce coherent attacks to collective by exploiting the symmetry of the states.

In \textbf{step 4} we show that the prover maximises the value of the bound of failure probability if they perform an attack with exactly the fewest non-trivial attacks that are consistent with $E_i$ obtained from step 2. This is a single attack for the protocol of this section (but different in Section \ref{Sec:FT_Verification}). Moreover, since we have a convex combination of positive numbers, the greatest value is obtained for a single $\sigma$. In the next steps of the proof we find the maximum value of our bound for an attack corresponding to the single optimal (for the prover)  $\sigma$. Here our technique has deviated from \cite{fk} so we will give more details later (and in Appendix \ref{App:proof1}).

Finally, in \textbf{step 5} we use the partition of the qubits to sets $P_v$ and $A_e$.  It is important to note, that within each of those sets there is exactly one computation qubit and exactly one trap qubit. From previous steps we know that the bound of the failure probability is highest if the prover chooses to make a single non-trivial attack. This attack happens at a qubit that belongs to either some set $P_v$ or some set $A_e$. The probability of hitting a trap given a single set is clearly independent from the other free parameters corresponding to this qubit (but not the probability to detect it in general) and it goes as $1/|P_v|$ or $1/|A_e|$. This leads us to a bound for the failure probability $\epsilon=8/9$.

We now give some definitions taken from \cite{fk} in order o elaborate further on the above steps at places we significantly differ from \cite{fk}, while the full proof is given in Appendix \ref{App:proof1}.  The \emph{output density operator} of the protocol is $B_j(\nu)$ and is given by

\EQ{
B_j(\nu)=\Tr_B\left(\sum_b \ket{b+c_r}\bra{b}C_{\nu_C,b}\Omega\mathcal{P}((\otimes^B\ket{0}\bra{0}\otimes\ket{\Psi^{\nu,b}}\bra{\Psi^{\nu,b}})\mathcal{P}^\dagger\Omega^\dagger C_{\nu_C, b}^\dagger\ket{b}\bra{b+c_r}\right)
}
where we have the following definitions: The subscript $j$ of the operator $B$, corresponds to the strategy/deviation that the prover makes, and when $j=0$ is the honest run where there is no deviation (and thus the operator $\Omega=\mathbb{I}$). The index $\nu$, collectively denotes all the random choices made by the verifier, i.e. $x,r,\theta,d$ and the positions of the traps $T$ (where the latter depends on the trap-colouring of the dotted triple-graph). The $b$'s are the outcomes of the prover's measurement, $(c_r)_i=r_i$ for $i\notin T$ and $(c_r)_t=0$ for $t\in T$, the subscript $B$ denotes tracing-out the prover's private registers. $C_{\nu_C,b}$ is the Pauli operator acting on the quantum output, that maps the final outcome to the correct one depending on the choices of random variables $\nu_C$ and the computation branch $b$. $\mathcal{P}$ is the unitary corresponding to implementing honestly the protocol. $\Omega$ is the deviation of the prover and is identity in the honest run.  $\ket{\Psi^{\nu,b}}=\ket{M^\nu}\otimes_j \ket{\delta_j^b}$ is the initial state send by the verifier, that includes the quantum input and the $\ket{+_{\theta}}$ states which are jointly denoted as $\ket{M^\nu}$ and depend on the random choices, and the $\ket{\delta}$ registers correspond to the measurement angles (that depend on the branch of the computation $b$).

The probability of failure of the protocol is when the protocol returns ``Accept'' but the output is orthogonal to the honest ideal. This probability is given by

\EQ{
p_{\textrm{fail}}&=& \sum_\nu p(\nu)\Tr (P^\nu_{\textrm{incorrect}} B_j(\nu))
}
where 

\EQ{
P^\nu_{\textrm{incorrect}}&=&(\mathbb{I}-\ket{\Psi_{\textrm{ideal}}}\bra{\Psi_{\textrm{ideal}}})\otimes_{t\in T}\ket{\eta_t^{\nu_T}}\bra{\eta_t^{\nu_T}}\nonumber\\
&=&P_\bot \otimes_{t\in T}\ket{\eta_t^{\nu_T}}\bra{\eta_t^{\nu_T}}
}
is the projection into the wrong subspace (orthogonal space to the correct ideal state) while it still remains within the accept subspace (where the traps succeed).  The prover's attack can be written in terms of Kraus operators such that $\sum_k\chi_k^\dagger\chi_k=\mathbb{I}$ which are written in terms of strings of Pauli as $\chi_k=\sum_i a_{ik}\sigma_i$. We then follow \cite{fk} up to step 3 and obtain the expression (see Appendix \ref{App:proof1}):

\EQ{
p_{\mathrm{fail}}&\leq & \sum_k\sum_{i\in E_i}|\alpha_{ki}|^2\sum_T p(T) \prod_{t\in T}\left(\sum_{\theta_t,r_t}p(\theta_t)p(r_t)(\bra{\eta_t^{\nu_T}}\sigma_{i|t}\ket{\eta_t^{\nu_T}})^2\right)
}
Here we denote $\sigma_{i|\gamma}$ the single-qubit Pauli matrix corresponding to position $\gamma$ of the string of Pauli $\sigma_i$.
From $\sum_{ik}|a_{ik}|^2=1$ we conclude that $p_{\mathrm{fail}}$ is maximised when $|a_{ik}|=0$ for all $i\notin E_i$. Moreover, it is maximised if $|a_{ik}|=1$ for a single $\sigma_i$.

\EQ{\label{p_inc}p_{\mathrm{fail}}&\leq & \max_{i\in E_i} \sum_T p(T) \prod_{t\in T}\left(\sum_{\theta_t,r_t}p(\theta_t)p(r_t)(\bra{\eta_t^{\nu_T}}\sigma_{i|t}\ket{\eta_t^{\nu_T}})^2\right)
}
This consists of a product of non-negative terms that are less than unity only if there is a non-trivial attack. It follows that the expression is maximum for a $\sigma_i$ that has non-trivial terms in as few as possible positions. This happens when $\sigma_i$ is non-trivial for a single position $\beta$ and this position belongs to either the set $P_{v_\beta}$ or the set $A_{e_\beta}$ depending on whether the non-trivial attack is done on a qubit that belongs to a primary set or an added set. The set that the attack belongs, is a property that does \emph{not} depend on the secret choices of the verifier, as for example, the trap-colouring. We will use the notation $F_\beta$ to refer to a set that can be either $P_{v_\beta}$ or $A_{e_\beta}$. 

We then break $p(T)$ which is the probability of different trap configurations using the structure of the subsets $P_v,A_e$, i.e. $p(T)=p(t_1\in P_{v_1},t_2\in P_{v_2}, \cdots, t_k\in A_{e_1}, \cdots)$. Therefore, given a single attack on set $F_{\beta}$ we can sum over all the other sets (all the other positions do not appear/have explicit dependency in the expression) and obtain $\sum_T p(T)=\sum_{t_\beta\in F_{\beta}}\sum_{t\notin F_\beta}p(T)=\sum_{t_\beta\in F_\beta}p(t_\beta)$. We obtain

\EQ{
p_{\mathrm{fail}}&\leq&  \max_{i\in E_i} \sum_{t_\beta\in F_\beta}\sum_{\theta_{t_\beta},r_{t_\beta}} p(t_\beta) p(\theta_{t_\beta})p(r_{t_\beta})(\bra{\eta_{t_\beta}^{\nu}}\sigma_{i|t_\beta}\ket{\eta_{t_\beta}^{\nu}})^2.
}
It is important to note that $\sigma_{i|t_\beta}$ is identity if $\beta\neq t_\beta$ while it is non-trivial otherwise, therefore $(|F_\beta|-1)$ terms of the sum will be unity, while one term will be less than one\footnote{It turns out that the not-unity term, is zero for measured qubits, while it can be up to 1/2 for output qubits.}. The above expression depends on whether the set $F_\beta$ is output set, or in the case that is a measured set on whether it is a primary or added set. Since the prover chooses which set to attacks the bound will be the highest of these values. We consider separately each case (see Appendix \ref{App:proof1}) and we obtain the highest $p_{\textrm{fail}}$  when there is a single attack on an added qubit. This happens because the added sets have 9 qubits $|A_e|=9$ and the chance of missing the trap is better. Noting that added qubits are not output qubits we thus obtain

\EQ{
p_{\mathrm{fail}}&\leq& \frac1{16|A_{e_\beta}|} \sum_{t_\beta\in A_{e_\beta}}\sum_{\theta_{t_\beta},r_{t_\beta}}(\bra{\eta_{t_\beta}^{\nu}}\sigma_{i|t_\beta}\ket{\eta_{t_\beta}^{\nu}})^2.\nonumber\\
& \leq & \frac1{16|A_e|} \sum_{r_t}\left(8\cdot (|A_{e_\beta}|-1)+8\cdot(\bra{r_{t_\beta}}\sigma_{i|t_\beta}\ket{r_{t_\beta}})^2\right) \nonumber\\
& \leq & \left(\frac 89\right) 
}

\end{proof}

\section{Amplification of the probability of success}

In the previous section we gave a simple construction to directly obtain a verification protocol with constant failure probability $\epsilon$. However, a verification protocol is successful if the $\epsilon$ of the failure probability can be made arbitrarily small. There are two techniques that have been used to amplify the probability of success of a verification protocol. The first one is simpler both conceptually and in terms of experimental requirements, but applies only in the case that the output of the quantum computation performed is classical. The second applies for computations with quantum output as well. We will use both techniques and show that starting with the dotted triple-graph construction we obtain in both cases improvements.

\subsection{Amplification for classical output}\label{Sec:Repetition}

As we mentioned in the introduction, in the case that a quantum computation has a classical output (e.g. solving classical problems or sampling, etc) it is sufficient to have a protocol that is $\epsilon$-verifiable for \emph{any} $\epsilon<1$. The reason is because one can boost this $\epsilon$ and make it arbitrarily small by repeating the protocol sufficiently many times and accepting only when all repetitions agree. This results to an $\epsilon'=\epsilon^d$ which can be made as small as the security level required by choosing the number of repetitions $d$ suitably.

In a previous protocol for classical output in \cite{fk}, the brickwork state was used. The brickwork state can have a single trap and leads to an $\epsilon=\left(1-\frac1n\right)$. By employing the repetition technique (that was used in \cite{efk}) this $\epsilon$ can be boosted to an $\epsilon'=\left(1-\frac1n\right)^d$. However, we see that in order to achieve a constant security level, the number of repetitions needs to increase as the size of the computation increases. For large enough $n$ it is easy to see, that if the single repetition requires $n$ qubits, to maintain a fixed security level  $O(n^2)$ qubits are needed. This is not better than the quantum-output protocol given in \cite{fk} that used a ``dotted-complete graph''. It is worth noting, that even though the complexity scales in the same way as the full ``dotted-complete graph'' protocol of \cite{fk}, the repetition protocol has still a number of practical advantages (it is easier to implement, has smaller coefficient of the leading term and for simple enough computation can be experimentally realised). 

By using the dotted triple-graph construction we can obtain a repetition protocol where we only repeat a constant number of times (and the number of repetitions depends \emph{only} on the security level). It follows, that the dotted triple-graph repetition protocol requires only a linear number of qubits. Again, it does not give better complexity from the general protocol (that allows for quantum output) which we give in the next section. However, it has the same practical advantages as the repetition protocol of the brickwork state.

\begin{algorithm}[H]
\caption{Boosted Verifiable UBQC using dotted triple-graph for classical output}
 \label{prot:RepUBQC}
\begin{itemize}

\item Verifier chooses a computation that has classical and deterministic output (e.g. a decision problem).

\item Verifier chooses a number $d$, where $d=\frac{\log \epsilon}{\log (8/9)}$ and the desired security level is $\epsilon$. 

\item For each $i\in\{1,\cdots,d\}$ \textbf{Verifier follows steps of Protocol \ref{prot:AUBQC}, with random different choices of secret parameters.} If the verifier accepts the computation, they register the classical output as $O_i$ and store it.

\item If the verifier rejected at any single repetition of Protocol \ref{prot:AUBQC}, they reject the overall computation. If not, they compare the classical outputs $O_i$ and if all of them are identical, they accept this output as the output of the computation.
\end{itemize}
\end{algorithm}

\begin{theorem} (Verification)\label{verification2}
Protocol \ref{prot:RepUBQC} is $\left(\frac{8}{9}\right)^d$-verifiable where the output is classical and $d$ is the number of repetitions.
\end{theorem}

\begin{proof}
We have multiple repetitions and if all of them return the same output $O$, then the probability that this is not the correct output is bounded by the probability that \emph{all} repetitions failed (and resulted to the same deviation). Since the different repetitions have the same outcome it means that if a single of those repetitions is successful then the output $O$ is the correct output. From Theorem \ref{Verification1} we know that the probability that a single repetition fails, is $8/9$. Then the probability that all the $d$ repetitions fail is $\left(\frac89\right)^d$. 
\end{proof}

\noindent\textbf{Alternative construction for classical output.} In the case of classical output, there is an alternative construction to the dotted triple-graph that could decrease further the (linear) overhead. In particular, instead of having the dotted triple-graph $DT(G)$ one could consider three copies of the dotted base graph $D(G)$. We will name this the \emph{three dotted copies} construction. The one copy will be used for computation, while the other two for white traps (on primary vertices) and black traps (on added vertices). This construction is global, in the sense that the decision of which vertices are in which graph is made from the beginning and cannot be decided independently per base graph vertex $v_i$. It follows that the location of the traps is totally correlated globally and there is no way to amplify the success probability in the quantum output case. This is the reason we focused on the dotted triple-graph construction for the quantum output case. 
For the classical output however, the three dotted copies construction works.

\begin{algorithm}[H]
\caption{Boosted three dotted copies Verifiable UBQC for classical output}
 \label{prot:Rep3copies}
\begin{itemize}

\item Verifier chooses a computation that has classical and deterministic output (e.g. a decision problem).

\item Verifier chooses a number $d$, where $d=\frac{\log \epsilon}{\log (2/3)}$ and the desired security level is $\epsilon$. 

\item For each $i\in\{1,\cdots,d\}$ \textbf{Verifier follows steps of Protocol \ref{prot:AUBQC}, using three dotted-copies instead of dotted triple-graph and with random different choices of secret parameters for every run.} If the verifier accepts the computation, they register the classical output as $O_i$ and store it.

\item If the verifier rejected at any single repetition of the modified Protocol \ref{prot:AUBQC}, they reject the overall computation. If not, they compare the classical outputs $O_i$ and if all of them are identical, they accept this output as the output of the computation.
\end{itemize}
\end{algorithm}

\begin{theorem} (Three dotted-copies Verification)\label{verification_3copies}
Protocol \ref{prot:Rep3copies} is $\left(\frac{2}{3}\right)^d$-verifiable where the output is classical and $d$ is the number of repetitions.
\end{theorem}

\begin{proof}

Following the proof of Theorem \ref{Verification1}, at step 2 in order to corrupt the computation the prover needs to make at least one non-trivial attack. However, the prover is blind about which of the three graphs is the computation, white and black trap graph. Therefore, it has probability $1/3$ that the attack coincides with a trap graph of the same type of the attack location (i.e. if it attacks a primary vertex, then with probability $1/3$ the attack was on a qubit that belongs in the white graph, while if the attack was on an added vertex with probability $1/3$ the attack was on a qubit that belongs in the black graph). For classical output the non-trivial attacks are $\{X,Y\}$ only (all qubits are measured), and thus the attack is deterministically detected when it hits a trap, as $\bra{\eta_{t}^{\nu}}\sigma_{X/Y}\ket{\eta_{t}^{\nu}}=\bra{r_t}\sigma_{X/Y}\ket{r_t}=0$.
Therefore the probability of failure is $p_{\textrm{fail}} \leq 2/3$. 

To amplify this probability, we can simply repeat the protocol $d$ times, and if all classical outputs agree in all the runs then the probability that the computation was corrupted is bounded by $p_{\textrm{fail}}\leq \left(\frac23\right)^d$. Moreover, the number of qubits required per repetition, are $3|V(G)|+3|E(G)|\leq 3(1+c)|V(G)$. Both in terms of failure probability and in terms of the (linear in both cases) number of qubits per repetition, the three copies construction gives better result.

\end{proof}

\subsection{Amplification for quantum output}\label{Sec:FT_Verification}

We now turn to the general case, where the output of the computation can be quantum. Our main result is that our dotted triple-graph construction leads to an exponentially-secure verification protocol for quantum output with only linear overhead. Similar to \cite{fk} we will use a technique that exploits the possibility to encode the computation in a fault-tolerant way in order to amplify the probability of success of the protocol.  The particular size of the boosting achieved will depend on the fault-tolerant code that is used. Here we treat the protocol in full generality. 

The general idea is that the computation is encoded using a fault-tolerant encoding, while the traps are still single (non-encoded) qubits. Therefore, while a single error on a trap leads to a rejection of the computation, to corrupt the actual output of the computation many errors on computation qubits are required. The prover needs to simultaneously avoid hitting any single trap and in the same time hit many computation qubits in order to corrupt the output. 

\begin{algorithm}[H]
\caption{Boosted Verifiable UBQC for quantum output, using dotted triple-graph and Fault-Tolerant Encoding}
 \label{prot:FTUBQC}
\begin{itemize}

\item Verifier chooses a base graph $G$ and a measurement pattern $\mathbb{M}_{\mathbb{\textrm{Comp}}}$ on the dotted base graph $D(G)$ that implements the desired computation in a fault-tolerant way, that can detect or correct errors fewer than $\delta/2$.

\item \textbf{Verifier follows steps of Protocol \ref{prot:AUBQC}.}
\end{itemize}
\end{algorithm}

\begin{theorem} (Verification)\label{verification3}
Protocol \ref{prot:FTUBQC} is $\left(\frac{8}{9}\right)^d$-verifiable for quantum or classical output, where $d=\lceil\frac\delta{2(2c+1)}\rceil$,  $c$ is the maximum degree of the base graph and $\delta$ is the number of errors tolerated on the base graph $G$.
\end{theorem}

\begin{proof}
We assume that there is a fault-tolerant encoding of the computation, that when done in MBQC, corrects or detects all errors that have fewer than $\delta$ number of errors when the computation is performed on the base graph $G$. Any operation on a measured qubit, diagonal in the computational basis ($\sigma_i\in\{I,Z\}$) does not alter the computation. Therefore errors that can contribute to corrupting a single logical qubit, involve errors $\sigma_i\in\{X,Y\}$ for measured qubits and $\sigma_i\in\{X,Y,Z\}$ for output qubits. 

Considering the dotted base graph $D(G)$, one can easily see that any (non-trivial) error on an added qubit $a_{e_{ij}}$, is equivalent with a local error on each of the two primary qubits that are neighbours $p_{v_i},p_{v_j}$ (see also \cite{fk}). If to corrupt a computation one needs $\delta$ errors on primary qubits of the base graph $G$, it follows that to corrupt the computation when done on the dotted base graph $D(G)$ one needs at least $\delta/2$ errors on qubits of the dotted base graph $D(G)$.

We now turn back to step 2 of the proof of theorem \ref{Verification1} and we see that the set $E_i$ of attacks that could possibly corrupt the computation, should include non-trivial attack in at least $\delta/2$ \emph{different} sets $P_v,A_e$ (which collectively we call $F_\beta$). It is important to note that within each of the sets $F_\beta$ there is a single computation-graph qubit and therefore not only the prover needs to perform $\delta/2$ non-trivial attacks, but they should also be done on at least $\delta/2$ different location sets. The prover of course could choose to attack multiple qubits of the same set $F_\beta$, but in order to hit at least $\delta/2$ computation qubits, the sets that they perform non-trivial attacks should also be at least $\delta/2$.  

Any given attack $\sigma_i$ is characterised by the set $S_i$ of locations on the dotted base graph $D(G)$, that it has at least one non-trivial attack, which in the case $\sigma_i\in E_i$ it means that $|S_i|\geq\delta/2$ .

Following step 3 and 4 of the proof of theorem \ref{Verification1},  we reach eq. (\ref{p_inc}). From this expression we can again see, that the fewer the positions of non-trivial attack (consistent with $E_i$) the greatest the value of this bound is. We already know that we need at least $\delta/2$ sets $F_\beta$ with non-trivial attacks, so it follows that the maximum is achieved when there are exactly $\delta/2$ different sets $F_\beta$ with exactly a single attack in each.

To proceed further we need to decompose the probability of different configuration of traps $p(T)$ to this of individual sets. This is not in general possible since there are correlation between the traps of (neighbouring) sets. To this point we should note that fixing a configuration of traps is identical with giving a trap-colouring as in definition \ref{trap colouring}. 

From Theorem \ref{ICL2}, we know that given a collection $S_i$ of $\delta/2$ locations on the dotted base graph $D(G)$ there are at least a collection $S'_i$ of $\frac\delta{2(2c+1)}$ that are independently colourable locations, i.e. $|S_i'|=\lceil\frac\delta{2(2c+1)}\rceil$. To obtain an upper bound on the failure probability, we set $\sigma_{i|\gamma}=I$ for all $\gamma$ that belong to locations in $S_i\setminus S'_i$. This change is non-decreasing for the expression for the bound given in eq. (\ref{p_inc}). Now the only locations that have non-trivial attacks, are those in $S_i'$ and we have

\EQ{\label{p_inc2}
p_{\mathrm{fail}}&\leq&  \max_{i\in E_i} \prod_{\beta=1}^{|S_i'|}\sum_{t_\beta\in F_\beta}p(t_\beta)\sum_{\theta_{t_\beta},r_{t_\beta}}  p(\theta_{t_\beta})p(r_{t_\beta})(\bra{\eta_{t_\beta}^{\nu}}\sigma_{i|t_\beta}\ket{\eta_{t_\beta}^{\nu}})^2
}
where we used the fact that $p(T)=p(t_1\in P_{v_1},t_2\in P_{v_2}, \cdots, t_k\in A_{e_1}, \cdots)$ and for a collection $S'_i$ of independent colourable locations, 

\EQ{
\sum_{t_\beta\in  S_i'}\left(\sum_{t_\beta\notin   S_i'}p(T)\right)=\sum_{t_\beta\in  S_i'}\left(\prod_{\beta=1}^{|S_i'|}p(t_\beta)\right)
}
We can see this due to the fact that after summing all locations apart from those in $S'$, the probability for choosing the location of the trap within each set is independent and thus the joint probability is simply their product. Now, each term in the product of eq. (\ref{p_inc2}) is bounded by $8/9$ as proven in the previous section and therefore we obtain

\EQ{\label{ft_bound}
p_{\mathrm{fail}}&\leq&\left(\frac{8}{9}\right)^d
}
where we define $d=|S_i'|=\lceil\frac\delta{2(2c+1)}\rceil$
\end{proof}
As a final remark, we should note that the value $d=\lceil\frac\delta{2(2c+1)}\rceil$ is the minimum number of ICL that exist and in particular cases this number can be greater and thus the probability of success of the verification protocol also becomes greater for those cases.

\section{Consequences for existing verification protocols}\label{Sec:Examples}

The dotted triple-graph construction can be used to improve a number of existing verification protocols and here we give indicatively three of them. First we consider the specific case where the computation is done using the Raussedorf-Harrington-Goyal (RHG) \cite{rhg} encoding and the related graph is $\mathcal{G}_{\mathcal{L}}$. Following our construction instead of the dotted-complete graph of \cite{fk}, we obtain the dotted triple-RHG graph $DT\mathcal{G}_{\mathcal{L}}$. This graph state has linear number of qubits (as the maximum degree of the graph is $4$). With the same choices of parameters as in \cite{fk}, it can detect or corrects any deviation that has fewer than $\delta/2$ errors. From our results of the previous section, it follows that we obtain a linear-complexity verification protocol with exponential security bound given by $p_{\mathrm{fail}}\leq\left(\frac89\right)^{\lceil\frac\delta{18}\rceil}$.

The second application is that it can be used to improve verifiable fault-tolerant protocols. Assuming that there are errors  due to noise (non-adversarial), the protocols given earlier in the text and in other VBQC  protocols could face a problem. In particular, honest errors due to noise could make trap measurement to fail and lead us to reject the output even in honest runs where the computation is not corrupted. Here we should stress that both in this paper and in \cite{fk} the use of fault-tolerant encoding was in order to amplify the security and \emph{not} to correct the computation from errors caused by honest noise.  However, one \emph{can} construct a fault-tolerant verification protocol, at least for classical output, and one such example is presented in \cite{gkw2015}. The starting graph used to obtain the fault-tolerant protocol of \cite{gkw2015} was the brickwork graph that has a single trap. Then a fault-tolerant encoding was done followed by the repetition technique used to amplify the success probability. However the number of repetitions to maintain a constant level of security increased with the size of the computation. By using the dotted triple-brickwork instead, as the first step of the construction in \cite{gkw2015}, we can achieve exponential security with a constant number of repetitions. This would essentially bring down the number of qubits required from $O(n^2)$ to $O(n)$. 

The third application is that we can directly use the dotted triple-graph construction for the verifiable measurement-only protocols \cite{tomo2014,hm2015}. In particular \cite{tomo2014} is essentially the online version of \cite{fk}, and the technique to include traps in the graph is equivalent. It follows that if the resource used instead of a dotted-complete graph is a dotted triple-RHG graph then the number of qubits required will reduce from quadratic to linear. In \cite{hm2015} the verifier, instead of including traps, uses $2k+1$ copies of a universal graph. In order to test the honesty of the prover it makes stabiliser measurements to $2k$ copies of the desired graph while performs the computation on the final copy. However, there is always at least a $1/(2k+1)$ probability that the computation is corrupted and not detected. Using the dotted triple-graph construction, modified for the measurement-only protocols, this probability can be made exponentially small while still using only linear number of qubits.

\section*{Acknowledgements}
The authors would like to thank Vedran Dunjko, Alexandru Gheorghiu and Theodoros Kapourniotis for useful discussions. The authors are also grateful to Joe Fitzsimons for useful discussion regarding a robust fault tolerance scheme as proposed in Section 5. EK acknowledges funding through EPSRC grants EP/N003829/1 and EP/M013243/1.

\appendix

\section{Proof of Verification}\label{App:proof1}

Here we give the detailed proof of Theorem \ref{Verification1}. In some parts we follow closely the proof of Theorem 8 of \cite{fk}. The proof has five steps. In \textbf{step 1} we express the attack using Kraus operators and Pauli matrices, in \textbf{step 2} we show that in order to lie in the incorrect subspace, at least one non-trivial attack to one qubit (of the dotted triple-graph)  is required, and then we will replace the projection to the incorrect subspace with this restriction on the sum of allowed attacks. In \textbf{step 3} we will exploit the blindness of the prover to reduce the attack to Pauli attacks. In \textbf{step 4} we will show that the fewer the non-trivial attacks the greater the probability for the adversary, and thus we will restrict to the fewer allowed attacks (a single one). Finally, in \textbf{step 5} we will use a suitable partition of the qubits which will then leads to a constant bound for the $p_{\textrm{fail}}$.

\noindent\textbf{ Step 1}: First we note that after tracing out the prover's register, the unitary $\Omega$ becomes a completely positive trace preserving map (CPTP), and can be expressed in terms of the Kraus operators $\{\chi_k\}$, where $\sum_k\chi_k\chi_k^\dagger=\mathbb{I}$. Moreover we express each Kraus operator as linear combination of Pauli operators $\chi_k=\sum_i \alpha_{ki}\sigma_i$ and $\sum_{k,i}\alpha_{ki}\alpha_{ki}^*=1$. The matrix $\sigma_i$ is a tensor product of Pauli matrices, where if we want to specify the Pauli acting on qubit $\gamma$ we will denote it as $\sigma_{i|\gamma}$. We then get

\EQ{
p_{\textrm{fail}}&=&\sum_\nu p(\nu)\Tr (P^\nu_{\textrm{incorrect}}B_j(\nu)) \nonumber\\
&=& \sum_{b,i,j,k}\Tr\left(\sum_\nu p(\nu)\alpha_{ki}\alpha^*_{kj}(P_\bot\otimes_{t\in T}\ket{\eta^{\nu_T}_t}\bra{\eta^{\nu_T}_t})\ket{b+c_r}\bra{b}\right.\nonumber\\
& & \left. C_{\nu_C,b}\sigma_i \mathcal{P}\ket{\Psi^{\nu,b}}\bra{\Psi^{\nu,b}}\mathcal{P}^\dagger \sigma_j C^\dagger_{\nu_C,b}\ket{b}\bra{b+c_r}\right)
}

\noindent\textbf{ Step 2}: Again following \cite{fk}, we can see that only terms that satisfy
 
\EQ{
\Tr(P_\bot\sigma_i\mathcal{P}\ket{\Psi^{\nu,b}}\bra{\Psi^{\nu,b}}\mathcal{P}^\dagger\sigma_j)\neq 0
} 
contribute to the $p_{\textrm{fail}}$. The terms that obey this are those necessarily within those that $|B_i|+|C_i|+|D^O_i|\geq 1$, which we will denote as $i\in E_i$ (and similarly $j\in E_j$), where the sets are defined as:

\EQ{
A_i&=&\{\gamma\textrm{ s.t. }\sigma_{i|\gamma}=I \textrm{ and }\gamma\textrm{ qubit of the dotted triple-graph}\}\nonumber\\
B_i&=&\{\gamma\textrm{ s.t. }\sigma_{i|\gamma}=X \textrm{ and }\gamma\textrm{ qubit of the dotted triple-graph}\}\nonumber\\
C_i&=&\{\gamma\textrm{ s.t. }\sigma_{\beta|\gamma}=Y \textrm{ and }\gamma\textrm{ qubit of the dotted triple-graph}\}\nonumber\\
D_i&=&\{\gamma\textrm{ s.t. }\sigma_{\beta|\gamma}=Z \textrm{ and }\gamma\textrm{ qubit of the dotted triple-graph}\}
}
and the superscript $O$ denotes subset of those sets that the $\gamma$ is output qubit. In other words, to corrupt the computation one either needs to flip the outcome of a measured qubit, or make any Pauli (other than the identity) if the attack is on the quantum output.

We have now imposed that the attacks $\sigma_i$ that contribute have at least one non-trivial Pauli attack at a qubit of the DT(G). This is not a sufficient condition to corrupt the computation in general (and send it to the $P_\bot$ subspace), but is a necessary condition. To see this, we note that if we consider a $\sigma_i$ where $i\notin E_i$, then there is no choice of the secret parameters that would bring the state in the $P_\bot$ subspace. Here we take the worse-case scenario, where we assume that if there is some choice of secret parameters that a given attack could corrupt the computation, then we assume that it already is in the subspace $P_\bot$ and we only check what is the probability that this attack did not triggered any trap. For protocol \ref{prot:AUBQC} it is a single attack that could corrupt the computation. We then replace the projection on the $P_\bot$ subspace, with a restriction on the possible attacks, i.e. at the sum we only have terms corresponding to attacks that belong to the set $E_i$. Note, that if the computation was encoded in an fault-tolerant way (as is done in Section \ref{Sec:FT_Verification}), then the set $E_i$ requires greater number of non-trivial attacks. For now we take the more conservative view. 

We then obtain the following expression:

\EQ{
p_{\textrm{fail}}\leq \sum_{k,b'} \sum _{\nu} p(\nu) \Tr \left((\otimes_{t\in T}\ket {\eta_t^{\nu_T}} \bra {\eta_t^{\nu_T}}\otimes \ket{b'}\bra{b'}) \left( \sum_{i\in E_i} \alpha_{ki} \sigma_i\right) \mathcal{P}  \ket {\Psi^{\nu,b'}}\bra {\Psi^{\nu,b'}} \mathcal{P}^\dagger \left( \sum_{i\in E_i} \alpha_{ki} \sigma_i\right)^\dagger \right)
}
where $b'=\{b_i\}_{i\notin T}$ a substring of $b$ that excludes the value for the trap measurements (and we used that $\bra{\eta_t^{\nu^T}}\ket{b_t}=\delta_{\eta_t^{\nu_T},b_t}$).

\noindent\textbf{ Step 3}: The next step is to exploit the fact that summing over the secret parameters of the verifier result to the prover being blind, and show that the only attacks that contribute are Pauli attacks, i.e. attacks that $\sigma_{i|\gamma}=\sigma_{j|\gamma}$ for all $\gamma$. Summing over $\nu_C$ we obtain

\EQ{
p_{\textrm{fail}}= \sum_{k,\nu_T} \sum_{i \in E_i} \sum_{j \in E_j} \alpha_{ik} \alpha_{jk}^* p(\nu_T) \Tr \left(\otimes_{t\in T}\ket{\eta_t^{\nu_T}}\bra {\eta_t^{\nu_T}}\sigma_i \left(\otimes_{t\in T}\ket {\eta_t^{\nu_T}} \bra {\eta_t^{\nu_T}} \otimes_{t\in T} \ket{\delta_t}\bra{\delta_t} \otimes {\frac{I}{\text{Tr}(I)}} \right) \sigma_j \right)
}
As all Pauli matrices but the identity are traceless, all terms in the sum are zero unless $\sigma_{i|\gamma}=\sigma_{j|\gamma}$ apart from the case that $\gamma\in T$. Then we use the fact that $\sum_{\theta_t,r_t}\Tr (\bra{\eta^{\nu_T}_t}\sigma_i\ket{\eta^{\nu_T}_t}\bra{\eta^{\nu_T}_t}\sigma_j\ket{\eta^{\nu_T}_t})=0$ unless $\sigma_{i|t}=\sigma_{j|t}$ in the case that $t\in O$ and that for measured traps it suffices to sum over $r_t$, i.e. $\sum_{r_t}\Tr (\bra{\eta^{\nu_T}_t}\sigma_i\ket{\eta^{\nu_T}_t}\bra{\eta^{\nu_T}_t}\sigma_j\ket{\eta^{\nu_T}_t})=0$ unless $\sigma_{i|t}=\sigma_{j|t}$. We then conclude that only terms that contribute are those that $\sigma_i=\sigma_j$. We thus obtain:

\EQ{
p_{\mathrm{fail}}&\leq & \sum_k\sum_{i\in E_i}|\alpha_{ki}|^2\sum_T p(T) \prod_{t\in T}\left(\sum_{\theta_t,r_t}p(\theta_t)p(r_t)(\bra{\eta_t^{\nu_T}}\sigma_{i|t}\ket{\eta_t^{\nu_T}})^2\right)
}
where we broke the sum of $\nu_T$ to the choice of positions $T$, and the random choices of $\theta_t,r_t$, and we have taken the product of all those terms corresponding to the various white and black traps.

\noindent\textbf{ Step 4}: In this step, we will prove that to maximise the value of the bound of the probability of $p_{\textrm{fail}}$, the best strategy is to do the fewer number of attacks allowed by the constraint obtained at step 2, which in our case, is a single attack. Then at the next step we will bound this maximum value. We have

\EQ{
p_{\mathrm{fail}}&\leq & \sum_k\sum_{i\in E_i}|\alpha_{ki}|^2 f(i)
}
where $f(i):=\sum_T p(T) \prod_{t\in T}\left(\sum_{\theta_t,r_t}p(\theta_t)p(r_t)(\bra{\eta_t^{\nu_T}}\sigma_{i|t}\ket{\eta_t^{\nu_T}})^2\right)$. From $\sum_{ik}|a_{ik}|^2=1$ we conclude that $p_{\mathrm{fail}}$ is maximised when $|a_{ik}|=0$ for all $i\notin E_i$. Then we have a convex combination of values $f(i)$. Let $f(m)=\max_{i\in E_i} f(i)$, then and it follows that if this is maximum for the single value $m$, then by choosing $|a_{ik}|=0$ for all $i\neq m$ the bound for $p_{\mathrm{fail}}$ is maximised.

\EQ{p_{\mathrm{incorrect}}&\leq& f(m)=\max_{i\in E_i} f(i)\nonumber\\
&\leq& \max_{i\in E_i} \sum_T p(T) \prod_{t\in T}\left(\sum_{\theta_t,r_t}p(\theta_t)p(r_t)(\bra{\eta_t^{\nu_T}}\sigma_{i|t}\ket{\eta_t^{\nu_T}})^2\right)}
In other words, we obtain a bound by considering a single $\sigma_i$ that belongs to the set $E_i$ and maximises the expression we have. The following expression involves a product of positive numbers, that are all less or equal to unity:

\EQ{\label{term1}\prod_{t\in T}\left(\sum_{\theta_t,r_t}p(\theta_t)p(r_t)(\bra{\eta_t^{\nu_T}}\sigma_{i|t}\ket{\eta_t^{\nu_T}})^2\right)}
In particular we can see that the terms in the product of eq. (\ref{term1}) are unity for all trap positions that $\sigma_{i|t}$ is trivial, i.e. $\sigma_{i|k}\notin \{X,Y\}$ if $k$ is not output, or $\sigma_{i|k}\notin \{X,Y,Z\}$ if $k$ is an output qubit. It is clear that this expression is bigger the more terms containing trivial attacks on traps. In other words, if we have two possible attacks $\sigma_i$ and $\sigma_{i'}$, where for all $\gamma$ that $\sigma_{i|\gamma}$ is non-trivial it is equal to $\sigma_{i'|\gamma}$ (but there are $\gamma$ that $\sigma_{i'|\gamma}$ is non-trivial while $\sigma_{i|\gamma}$ is trivial), then $f(i)\geq f(i')$. Therefore the term that maximises $p_{\mathrm{fail}}$ corresponds to an attack $\sigma_i$ that has the fewest (possible, i.e. compatible with $E_i$) non-trivial terms.

From step 2 we obtained that the set $E_i$ has at least one non-trivial Pauli attack, so it follows that the bound of the $p_{\textrm{fail}}$ we compute is maximised when there is exactly one non-trivial Pauli attack. It is important to note however, that the set $E_i$ will be different in Section \ref{Sec:FT_Verification} where we consider fault-tolerant encoding of the computation and the corresponding $\sigma_i$ will involve greater number of non-trivial attacks. In that case, the set of attacks that can possibly corrupt the computation (and thus send it to $P_\bot$ subspace) changes (i.e. $E_i$ differs).

\noindent\textbf{ Step 5}: We will now use the partition of the qubits of the dotted triple-graph, to the subsets $P_v,A_e$ corresponding to vertices and edges of the base graph. The way that this partition is chosen does not reveal any new information to the prover and does not depend on the choice of trap-colouring, i.e. on the positions of the traps.

We have established that the optimal strategy for the prover in order to maximise the value of the bound for the $p_{\mathrm{fail}}$ we compute, is to make a single non-trivial attack at one qubit of the dotted triple-graph. Let us assume that this single position is $\beta$ and we know that it belongs to either a set $P_{v_\beta}$ or a set $A_{e_\beta}$ depending on whether the non-trivial attack is done on a qubit belonging to a primary set $P_{v_\beta}$ or an added set $A_{e_\beta}$.  When it is not clear if the set is primary or added, we will use $F_\beta$ which simply means that $F_\beta=P_{v_\beta}$ if $\beta$ is at a primary location and $F_\beta=A_{e_\beta}$ if $\beta$ is at an added location.  

We then break the $p(T)$ which is the probability of different trap configurations, using the structure of the subsets $P_v,A_e$, i.e. $p(T)=p(t_1\in P_{v_1},t_2\in P_{v_2}, \cdots, t_k\in A_{e_1}, \cdots)$. Therefore, given a single attack at set $F_{\beta}$, we can sum over all the other sets (all the other positions do not appear in the remaining expression) and obtain $\sum_T p(T)=\sum_{t_\beta\in F_{\beta}}\sum_{t\notin F_\beta}p(T)=\sum_{t_\beta\in F_\beta}p(t_\beta)$. We obtain

\EQ{
p_{\mathrm{fail}}&\leq&  \max_{i\in E_i} \sum_{t_\beta\in F_\beta}\sum_{\theta_{t_\beta},r_{t_\beta}} p(t_\beta) p(\theta_{t_\beta})p(r_{t_\beta})(\bra{\eta_{t_\beta}^{\nu}}\sigma_{i|t_\beta}\ket{\eta_{t_\beta}^{\nu}})^2.
}
It is important to note that $\sigma_{i|t_\beta}$ is identity if $\beta\neq t_\beta$ while it is non-trivial otherwise, therefore $(|F_\beta|-1)$ terms of the sum will be unity, while one term will be less than one\footnote{It turns out that the not-unity term, is zero for measured qubits, while it can be up to 1/2 for output qubits.}. The above expression depends on whether the set $F_\beta$ is output set, or in the case that is a measured set one on whether it is a primary or added set. It will be the prover that chooses which is the set of the attack, and thus the bound will be the highest of these values. We consider separately each case. We define the quantity 

\EQ{\label{def_g}
g(i,F_\beta)= \sum_{t_\beta\in F_{t_\beta}}\sum_{\theta_{t_\beta},r_{t_\beta}} p(t_\beta) p(\theta_{t_\beta})p(r_{t_\beta})(\bra{\eta_{t_\beta}^{\nu}}\sigma_{i|t_\beta}\ket{\eta_{t_\beta}^{\nu}})^2
} 
where the function $g$ has explicit the dependence on which set $F_\beta$ does the non-trivial attack belong to. In particular, we will denote $P_{v_\beta}^O$ if the non-trivial attack is on an output set (note that output qubits are only primary), and $P_{v_{\beta'}}$ if it is on a measured primary set and $A_{e_{\beta'}}$ if it is on a measured added set. We will separately compute the maximum of $g(i,P_{v_\beta}^O),g(i,P_{v_{\beta'}}),g(i,A_{e_{\beta'}})$ for $i\in E_i$ and the bound will be the maximum of those three.

We start with the output qubits

\EQ{
g(i,P_{v_\beta}^O)&=& \frac1{16|P_{v_\beta}^O|}\sum_{t_\beta\in P_{v_\beta}^O}\sum_{\theta_t,r_t}(\bra{+_{\theta}}\sigma_{i|t_\beta}\ket{+_\theta})^2 \nonumber\\
&=& \frac1{16\times 3}\sum_{\theta_t,r_t}\left(1\cdot (|P_{v_\beta}^O|-1)+1\cdot(\bra{+_{\theta}}\sigma_{i|t_\beta}\ket{+_\theta})^2\right) \nonumber\\
&\leq &  \frac1{16\times 3}\left(16\cdot (|P_{v_\beta}^O|-1)+8\right) \nonumber\\
&\leq & \left(1-\frac1{2|P_{v_\beta}^O|}\right)=\frac56 
}
where we used that $\sum_\theta (\bra{+_\theta}\sigma\ket{+_\theta})^2\leq 4$ for $\sigma\neq \mathbb{I}$ and that $|P_{v_\beta}^O|=3$ since output qubits are primary.

Now we consider the measured qubits similarly (using $F_{\beta'}$ to denote either primary or added measured set), to obtain 

\EQ{
g(i,F_{\beta'})&=& \frac1{16|F_{\beta'}|}\sum_{t_\beta\in F_{\beta'}}\sum_{\theta_t,r_t}(\bra{r_{t_\beta}}\sigma_{i|t_\beta}\ket{r_{t_\beta}})^2 \nonumber\\
&=& \frac1{16|F_{\beta'}|}\sum_{r_t}\left(8\cdot (|F_{\beta'}|-1)+8\cdot(\bra{r_{t_\beta}}\sigma_{i|t_\beta}\ket{r_{t_\beta}})^2\right) \nonumber\\
&= &  \frac1{16|F_{\beta'}|}\left(16\cdot (|F_{\beta'}|-1)\right) \nonumber\\
&= & \left(1-\frac1{|F_{\beta'}|}\right) \leq \frac 89
}
where the last step the equality holds if the attack is on added qubits where $|F_{\beta'}|=|A_{e_{\beta'}}|=9$. For primary measured sets $P_{v_{\beta'}}$ the bound is only $(2/3)$ and thus is lower. It follows that the overall bound we obtain (worse-case) is

\EQ{
p_{\textrm{fail}}\leq \left(\frac89\right)
}
Since this bound is obtained when the attack is on a measured (added) qubit, the bound is the same when the output is fully classical.

\bibliography{report}
\bibliographystyle{unsrt}

\end{document}